\documentclass{llncs}

\usepackage{amssymb}
\usepackage[intlimits]{amsmath}
\usepackage{amsfonts}
\usepackage{enumerate}

\renewcommand{\qed}{\rule{1ex}{1ex}}
\newcommand{\delim}{\left.\right|}

\title{A Tight Lower Bound on Certificate Complexity in Terms of Block Sensitivity and Sensitivity\thanks{
This research has received funding from the EU Seventh Framework Programme (FP7/2007-2013) under projects QALGO (No. 600700) and RAQUEL (No. 323970)
and ERC Advanced Grant MQC. 
Part of this work was done while Andris Ambainis was
visiting Institute for Advanced Study, Princeton, supported by National Science 
Foundation under agreement No. DMS-1128155. Any opinions, findings and conclusions or recommendations expressed in this material are those of the author(s) and 
do not necessarily reflect the views of the National Science Foundation.}}
\author{Andris Ambainis \and Kri\v{s}j\=anis Pr\=usis}
\institute{Faculty of Computing, University of Latvia, Raina bulv. 19, R\=\i ga, LV-1586, Latvia}

\begin{document}

\maketitle

\begin{abstract}
Sensitivity, certificate complexity and block sensitivity are widely used Boolean function complexity measures. A longstanding open problem, proposed by Nisan and Szegedy \cite{NS}, is whether sensitivity and block sensitivity are polynomially related. Motivated by the constructions of functions which achieve the largest known separations, we study the relation between 1-certificate complexity and 0-sensitivity and 0-block sensitivity. 

Previously the best known lower bound was  $C_1(f)\geq \frac{bs_0(f)}{2 s_0(f)}$, achieved by Kenyon and Kutin \cite{KK}. We improve this to $C_1(f)\geq \frac{3 bs_0(f)}{2 s_0(f)}$. While this improvement is only by a constant factor, this is quite important, as it precludes achieving a superquadratic separation between $bs(f)$ and $s(f)$ by iterating functions which reach this bound. In addition, this bound is tight, as it matches the construction of Ambainis and Sun \cite{AS} up to an additive constant.
\end{abstract}

\section{Introduction}

Determining the biggest possible gap between the sensitivity $s(f)$ and block sensitivity $bs(f)$ of a Boolean function is a well-known open
problem in the complexity of Boolean functions. Even though this question has been known for over 20 years, there has been quite little progress on it.

The biggest known gap is $bs(f)=\Omega(s^2(f))$. This was first discovered by Rubinstein \cite{R}, who constructed a function $f$ with $bs(f)=\frac{s^2(f)}{2}$,
and then improved by Virza \cite{V} and Ambainis and Sun \cite{AS}. Currently, the best result is 
a function $f$ with $bs(f)=\frac{2}{3} s^2(f) - \frac{1}{3} s(f)$ \cite{AS}.
The best known upper bound is exponential: $bs(f)\leq s(f) 2^{s(f)-1}$ \cite{A+} which improves over an earlier exponential upper bound by Kenyon and Kutin \cite{KK}.

In this paper, we study a question motivated by the constructions of functions that achieve a separation between $s(f)$ and $bs(f)$. The question is as follows:
Let $s_z(f)$, $bs_z(f)$ and $C_{z}(f)$ be the maximum sensitivity, block sensitivity and certificate complexity achieved by $f$ on inputs $x$: $f(x)=z$. 
What is the best lower bound of $C_1(f)$ in terms of $s_0(f)$ and $bs_0(f)$?

The motivation for this question is as follows. Assume that we fix $s_0(f)$ to a relatively small value $m$ and fix $bs_0(f)$ to a substantially larger value $k$. 
We then minimize $C_1(f)$. We know that $s_1(f)\leq C_1(f)$ (because every sensitive bit has to be contained in a certificate). 
We have now constructed an example where both $s_0(f)$ and $s_1(f)$ are relatively small and $bs_0(f)$ large.
This may already achieve a separation between $bs_0(f)$ and $s(f)=\max(s_0(f), s_1(f))$ and, if $s_1(f)>s_0(f)$, we can improve this separation by composing the function with OR (as in \cite{AS}).

While this is just one way of achieving a gap between $s(f)$ and $bs(f)$, all the best separations between these two quantities can be cast into this
framework. Therefore, we think that it is interesting to explore the limits of this approach.

The previous results are as follows:
\begin{enumerate}
\item
Rubinstein's construction \cite{R} can be viewed as taking a function $f$ with $s_0(f)=1$, $bs_0(f)=k$ and $C_1(f)=2k$. A composition with OR yields \cite{AS}
$bs(f)=\frac{1}{2}s^2(f)$;
\item
Later work by Virza \cite{V} and Ambainis and Sun \cite{AS} improves this construction by constructing $f$ with $s_0(f)=1$, $bs_0(f)=k$ and 
$C_1(f)=\left\lfloor \frac{3k}{2} \right\rfloor + 1$. A composition with OR yields $bs(f)=\frac{2}{3} s^2(f) - \frac{1}{3} s(f)$;
\item
Ambainis and Sun \cite{AS} also show that, given $s_0(f)=1$ and $bs_0(f)=k$, the certificate complexity $C_1(f)=\left\lfloor \frac{3k}{2} \right\rfloor + 1$
is the smallest that can be achieved. This means that a better bound must either start with $f$ with $s_0(f)>1$ or use some other approach;
\item
For $s_0(f)=m$ and $bs_0(f)=k$, it is easy to modify the construction of Ambainis and Sun \cite{AS} to obtain
$C_1(f)=\left\lfloor \frac{3\lceil k/m\rceil}{2} \right\rfloor + 1$ but this does not result in a better separation between $bs(f)$ and $s(f)$;
\item
Kenyon and Kutin \cite{KK} have shown a lower bound of $C_1(f)\geq \frac{k}{2m}$. If this was achievable, this could result in a separation of 
$bs(f)=2 s^2(f)$.
\end{enumerate}
The gap between the construction $C_1(f) = \frac{3k}{2m}+O(1)$ and the lower bound of $C_1(f)\geq \frac{k}{2m}$ is only a constant factor
but the constant here is quite important. This gap corresponds to a difference between $bs(f)=(\frac{2}{3}+o(1)) s^2(f)$ and $bs(f)=2s^2(f)$,
and, if we achieved $bs(f)>s^2(f)$, iterating the function $f$ would yield an infinite sequence of functions with a 
superquadratic separation $bs(f)=s(f)^c$, where $c>2$.

In this paper, we show that for any $f$
\[ C_1(f) \geq \frac{3}{2} \frac{bs_0(f)}{s_0(f)} - \frac{1}{2} .\]
This matches the best construction up to an additive constant and shows that no further improvement can be achieved along the lines of \cite{R,V,AS}.
Our bound is shown by an intricate analysis of possible certificate structures for $f$. 

Since we now know that $bs_0(f) \leq \left(\frac{2}{3}+o(1)\right)C_1(f) s_0(f)$, it is tempting to conjecture that $bs_0(f) \leq \left(\frac{2}{3}+o(1)\right)s_1(f) s_0(f)$. 
If this was true, the existing separation between $bs(f)$ and $s(f)$ would be tight.

\section{Preliminaries}

Let $f: \{0,1\}^n \rightarrow \{0,1\}$ be a Boolean function on $n$ variables. The $i$-th variable of input $x$ is denoted by $x_i$. For an index set $S \subseteq [n]$, let $x^S$ be the input obtained from an input $x$ by flipping every bit $x_i$, $i \in S$. Let a \emph{$z$-input} be an input on which the function takes the value $z$, where $z \in \{0,1 \}$.

We briefly define the notions of sensitivity, block sensitivity and certificate complexity.
For more information on them and their relations to other
complexity measures (such as deterministic, probabilistic and quantum decision
tree complexities), we refer the reader to the surveys by Buhrman and de Wolf \cite{BW}
and Hatami et al. \cite{HKP}.

\begin{definition}
The \emph{sensitivity complexity} $s(f,x)$ of $f$ on an input $x$ is defined as $| \{ i \delim f(x) \neq f(x^{\{i\}})\} |$. The \emph{$z$-sensitivity} $s_z(f)$ of $f$, where $z \in \{0,1 \}$,  is defined as $\max \{s(f,x) \delim x \in \{0,1\}^n, f(x)=z\}$.  The \emph{sensitivity} $s(f)$ of $f$  is defined as $\max \{s_0(f),s_1(f)\}$.
\end{definition}

\begin{definition}
The \emph{block sensitivity} $bs(f,x)$ of $f$ on input $x$ is defined as the maximum number $b$ such that there are $b$ pairwise disjoint subsets $B_1, \ldots , B_b$ of $[n]$ for which $f(x) \neq f(x^{B_i})$. We call each $B_i$ a \emph{block}.  The \emph{$z$-block sensitivity} $bs_z(f)$ of $f$, where $z \in \{0,1 \}$,  is defined as $\max \{bs(f,x) \delim x \in \{0,1\}^n, f(x)=z\}$.  The \emph{block sensitivity} $bs(f)$ of $f$  is defined as $\max \{bs_0(f),bs_1(f)\}$.
\end{definition}

\begin{definition}
A \emph{certificate} $c$ of $f$ on input $x$ is defined as a partial assignment $c: S \rightarrow \{0,1\}, S \subseteq [n]$ of $x$ such that $f$ is constant on this restriction. If $f$ is always 0 on this restriction, the certificate is a \emph{0-certificate}. If $f$ is always 1, the certificate is a \emph{1-certificate}.
\end{definition}

We denote specific certificates as words with $*$ in the positions that the certificate does not assign. For example, $01\!*\!*\!*\!*$ denotes a certificate that assigns 0 to the first variable and 1 to the second variable.

We say that an input $x$ \emph{satisfies} a certificate $c$ if it matches the certificate in every assigned bit. 

The number of \emph{contradictions} between an input and a certificate or between two certificates is the number of positions where one of them assigns 1 and the other assigns 0. For example, there are two contradictions between $0010\!*\!*$ and $100\!*\!**$ (in the 1st position and the 3rd position).

The number of \emph{overlaps} between two certificates is the number of positions where both have assigned the same values. For example, there is one overlap between $001\!*\!**$ and $*0000$ (in the second position). We say that two certificates {\em overlap} if there is at least one overlap between them.

We say that a certificate remains \emph{valid} after fixing some input bits if none of the fixed bits contradicts the certificate's assignments.

\begin{definition}
The \emph{certificate complexity} $C(f,x)$ of $f$ on input $x$ is defined as the minimum length of a certificate that $x$ satisfies. The \emph{$z$-certificate complexity} $C_z(f)$ of $f$, where $z \in \{0,1 \}$,  is defined as $\max \{C(f,x) \delim x \in \{0,1\}^n, f(x)=z\}$.  The \emph{certificate complexity} $C(f)$ of $f$  is defined as $\max \{C_0(f),C_1(f)\}$.
\end{definition}

\section{Background}

We study the following question:

{\bf Question:} Assume that $s_0(g)=m$ and $bs_0(g)=k$. How small can we make $C_1(g)$?

{\bf Example 1.} Ambainis and Sun \cite{AS} consider the following construction.

They define $g_0(x_1, \ldots, x_{2k})=1$ if and only if 
$(x_1, \ldots, x_{2k})$ satisfies one of $k$ certificates 
$c_0, \ldots, c_{k-1}$ with 
$c_i$  ($i\in\{0, 1, \ldots, k-1\}$) requiring that
\begin{enumerate}
\item[(a)]
$x_{2i+1}=x_{2i+2}=1$;
\item[(b)]
$x_{2j+1}=0$ for $j\in\{0, \ldots, k-1\}$, $j\neq i$;
\item[(c)]
$x_{2j+2}=0$ for $j\in\{i+1, \ldots, i+\lfloor k/2\rfloor\}$ (with $i+1, \ldots, i+\lfloor k/2\rfloor$
taken $\bmod k$).
\end{enumerate}
Then, we have:
\begin{itemize} 
\item
$s_0(g_0)=1$ (it can be shown that, for every 0-input of $g_0$, there is at most one $c_i$ in which only one variable does not have the right value);
\item
$s_1(g_0)=C_1(g_0)=\lfloor 3k/2\rfloor+1$ (a 1-input that satisfies a certificate $c_i$ is
sensitive to changing any of the variables in $c_i$ and $c_i$ contains 
$\lfloor 3k/2\rfloor+1$ variables);
\item
$bs_0(g_0)=k$ (the 0-input $x_1=\cdots=x_{2k}=0$ is sensitive to changing any of the pairs $(x_{2i+1}, x_{2i+2})$ from $(0, 0)$ to $(1, 1)$).
\end{itemize}

This function can be composed with the OR-function to obtain the best known separation between $s(f)$ and $bs(f)$: $bs(f)=\frac{2}{3} s^2(f) - \frac{1}{3} s(f)$\cite{AS}. As long as $s_0(g)=1$, the construction is essentially optimal: any $g$ with $bs_0(g)=k$ must satisfy $C_1(g)\geq s_1(g) \geq \frac{3k}{2}-O(1)$.

In this paper, we explore the case when $s_0(g)>1$.
An easy modification of the construction from \cite{AS} gives

\begin{theorem}
\label{thm:easy}
There exists a function $g$ for which
$s_0(g)=m$, $bs_0(g)=k$ and 
$C_1(f)=\left\lfloor \frac{3\lceil k/m\rceil}{2} \right\rfloor + 1$.
\end{theorem}

\begin{proof}
To simplify the notation, we assume that $k$ is divisible by $m$. Let $r=k/m$.

We consider a function $g(x_{m1}, \ldots, x_{m,2r})$ with variables $x_{i, j}$ (
$i\in\{1, \ldots, m\}$ and $j\in\{1, \ldots, 2r\}$) defined by
\begin{equation}
\label{eq:or} g(x_{11}, \ldots, x_{m,2r}) = \vee_{i=1}^m g_0(x_{i,1}, \ldots, x_{i,2r}) .
\end{equation}
Equivalently, $g(x_{11}, \ldots, x_{m,2r})=1$ if and only if at least one of 
the blocks $(x_{i,1}, \ldots, x_{i,2r})$ satisfies one of the certificates 
$c_{i,0} , \ldots , c_{i,r-1}$ that are defined similarly to $c_0 , \ldots , c_{k-1}$
in the definition of $g_0$. 

It is easy to see \cite{AS} that composing a function $g_0$ with OR gives
$s_0(g)=m\, s_0(g_0)$, $bs_0(g)=m\, bs_0(g_0)$ and
$C_1(g) = C_1(g_0)$, implying the theorem.
\qed
\end{proof}

While this function does not give a better separation between $s(f)$ and $bs(f)$, 
any improvement to Theorem \ref{thm:easy} could 
give a better separation between $s(f)$ and $bs(f)$ by using the same composition with OR
as in \cite{AS}.

On the other hand, Kenyon and Kutin \cite{KK} have shown that

\begin{theorem}
For any $f$ with $s_0(g)=m$ and $bs_0(g)=k$, we have $C_1(f)\geq \frac{k}{2m}$.
\end{theorem}

\section{Separation between $C_1(f)$ and $bs_0(f)$}

In this paper, we show that the example of Theorem \ref{thm:easy} is
optimal.

\begin{theorem} \label {theorem:result}
 For any Boolean function $f$ the following inequality holds:
\begin{equation} \label {equation:result}
 C_1(f) \geq \frac{3}{2} \frac{bs_0(f)} {s_0(f)} - \frac{1}{2}.
\end{equation}
\end{theorem}

\begin{proof}
Without loss of generality, we can assume that the maximum $bs_0$ is achieved on the all-0 input denoted by 0. Let $B_1,...,B_k$  be the sensitive blocks, where $k=bs_0(f)$.
Also, we can w.l.o.g. assume that these blocks are minimal and that every bit belongs to a block. (Otherwise, we can fix the remaining bits to 0.
This can only decrease $s_0$ and $C_1$, strengthening the result.)

Each block $B_i$ has a corresponding minimal 1-certificate $c_i$ such that the word $(\{0\}^n)^{B_i}$ satisfies this certificate. Each of these certificates has a 1 in every position of the corresponding block (otherwise the block would not be minimal) and any number of 0's in other blocks.

We construct a complete weighted graph $G$ whose vertices correspond to certificates $c_1$, $\ldots$, $c_k$.
Each edge has a weight that is equal to the number of contradictions between the two certificates the edge connects. 
{\em The weight of a graph} is just the sum of the weights of its edges. We will prove

\begin{lemma}
Let $w$ be the weight of an induced subgraph of $G$ of order $m$. Then
\begin{equation}
w \geq \frac{3}{2} \frac{m^2}{s_0(f)}-\frac{3}{2}m.
\end{equation}
\end{lemma}

\begin{proof}
The proof is by induction.
As a basis we take induced subgraphs of order $m \leq s_0(f)$. In this case,

\begin{equation}
\frac{3}{2} \frac{m^2}{s_0(f)}-\frac{3}{2}m \leq 0
\end{equation}
and $w \geq 0$ is always true, as the number of contradictions between two certificates cannot be negative.

Let $m > s_0(f)$.
We assume that the relation holds for every induced subgraph of order $< m$. Let $G'$ be an induced subgraph of order $m$. 
Let $H \subset G'$ be its induced subgraph of order $s_0(f)$ with the smallest total weight.

\begin{lemma} \label{lemma:subgraphs}
For any certificate $c_i \in G' \setminus H$  in $G'$ not belonging to this subgraph $H$ the weight of the edges connecting $c_i$ to $H$ is $\geq 3$.
\end {lemma}
\begin {proof}
Let $t$ be the total weight of the edges in $H$. Let us assume that there exists a certificate $c_j \notin H$ such that the weight of the edges connecting $c_j$ to $H$ is $\leq 2$. Let $H'$ be the induced subgraph $H \cup \{c_j\}$. Then the weight of $H'$ must be $\leq t+2$.

We define the weight of a certificate $c_i \in H'$ as the sum of the weights of all edges of $H'$ that involve vertex $c_i$. 
If there exists a certificate $c_i \in H'$ such that its weight in $H'$ is $\geq 3$, then the weight of $H' \setminus \{c_i\}$ would be $<t$, which is a contradiction, as $H$ was taken to be the induced subgraph of order $s_0(f)$ with the smallest weight. Therefore the weight of every certificate in $H'$ is at most 2.

In the next section, we show
\begin{lemma} \label{lemma:overlaps}
Let $f$ be a Boolean function for which the following properties hold: $f(\{0\}^n)=0$ and $f$ has such $k$ minimal 1-certificates that each has at most 2 contradictions with all the others together. Furthermore, for each input position, exactly one of these certificates assigns the value 1. Then, $s_0(f) \geq k$.
\end{lemma}


This lemma implies that $s_0(f)\geq |H'|$ which is in contradiction with $|H'|=s_0(f)+1$. Therefore no such $c_j$ exists.
\qed
\end{proof}

We now examine the graph $G' \setminus H$. It consists of $m-s_0(f)$ certificates and by the inductive assumption has a weight of at least 
\begin{equation}
\frac{3}{2} \frac{(m-s_0(f))^2}{s_0(f)} - \frac{3}{2} (m-s_0(f)).
\end{equation}
But there are at least $3 (m-s_0(f))$ contradictions between $H$ and $G' \setminus H$, thus the total weight of $G'$ is at least
\begin{align}
&\frac{3}{2} \frac{(m-s_0(f))^2}{s_0(f)} - \frac{3}{2} (m-s_0(f)) + 3(m-s_0(f))  \\
=~&\frac{3}{2} \frac{m^2-2 m s_0(f) + s_0(f)^2 } {s_0(f)} + \frac{3}{2} m - \frac{3}{2} s_0(f) \\
=~&\frac{3}{2} \frac{m^2} {s_0(f)} - 3 m + \frac{3}{2} s_0(f) + \frac{3}{2} m - \frac{3}{2} s_0(f) \\
=~&\frac{3}{2} \frac{m^2}{s_0(f)} - \frac{3}{2} m .
\end{align}
This completes the induction step. 
\qed
\end{proof}

By taking the whole of $G$ as $G'$, we find a lower bound on the total number of contradictions in the graph:
\begin{equation}
\frac{3}{2} \frac{k^2}{s_0(f)} - \frac{3}{2} k .
\end{equation}
Each contradiction requires one 0 in one of the certificates and each 0 contributes to exactly one contradiction (since for each position exactly one of $c_i$ assigns  a 1). Therefore, by the pigeonhole principle, there exists a certificate with at least 
\begin{equation}
\frac{3}{2} \frac{k}{s_0(f)} - \frac{3}{2}
\end{equation}
zeroes. As each certificate contains at least one 1, we get a lower bound on the size of one of these certificates and $C_1(f)$:
\begin{equation}
C_1(f) \geq \frac{3}{2} \frac{bs_0(f)}{s_0(f)} - \frac{1}{2}.
\end{equation}
\qed
\end{proof}

\section{Functions with $s_0(f)$ Equal to Number of 1-certificates}

In this section we prove Lemma \ref{lemma:overlaps}.

\subsection{General Case: Functions with Overlaps} \label{section:overlaps}

Let $c_1, \ldots, c_k$ be the $k$ certificates.
We start by reducing the general case of Lemma \ref{lemma:overlaps} to the case when there are no overlaps between any of $c_1, \ldots, c_k$.


Note that certificate overlaps can only occur when two certificates assign 0 to the same position. Then a third certificate assigns 1 to that position. This produces 2 contradictions for the third certificate, therefore it has no further overlaps or contradictions. For example, here we have this situation in the 3rd position (with the first three certificates) and in the 6th position (with the last three certificates):

\begin{equation}
\begin{pmatrix}
1&1&0&*&*&*&*&*&*&*\\
*&*&1&*&*&*&*&*&*&*\\
*&*&0&1&1&0&*&*&*&*\\
*&*&*&*&*&1&1&1&*&*\\
0&*&*&*&*&0&*&*&1&1
\end{pmatrix}.
\end{equation}

Let $t$ be the total number of such overlaps. Let $D$ be the set of certificates assigning 1 to positions with overlaps, $|D|=t$.
We fix the position of every overlap to 0. Since the remaining function contains the word $\{0\}^n$, it is not identically 1. Every certificate not in $D$ is still a valid 1-certificate, as they assigned either nothing or 0 to the fixed positions. If they are no longer minimal, we can minimize them, which cannot produce any new overlaps or contradictions.

The certificates in $D$ are, however, no longer valid. Let us examine one such certificate $c \in D$. We denote the set of positions assigned to by $c$ by $S$. 
Let $i$ be the position in $S$ that is now fixed to 0.
We claim that certificate $c$ assigns value 1 to all $|S|$ positions in $c$. 
(If it assigned 0 to some position, there would be at least 3 contradictions between
$c$ and other certificates: two in position $i$ and one in position where $c$ assigns 0.)

If $|S|=1$, then the remaining function is always sensitive to $i$ on 0-inputs, as flipping $x_i$ results in an input satisfying $c$. 

If $|S|>1$, we examine the $2^{|S|-1}$ subfunctions obtainable by fixing the remaining positions of $S$. We fix these positions to the subfunction that is not identically 1 with the highest number of bits fixed to 1, we will call this the \emph{largest non-constant subfunction}. If it fixes 1 in every position, it is sensitive to $i$ on 0-inputs, as flipping it produces a word which satisfies $c$. Otherwise it is sensitive on 0-inputs to every other bit fixed to 0 in $S$  besides $i$, as flipping them would produce a word from a subfunction with a higher amount of bits fixed to 1. But that subfunction is identically 1 or we would have fixed it instead.

In either case we obtain at least one sensitive bit in $S$ on 0-inputs in the remaining function. Furthermore, every certificate not in $D$ is still valid, if not minimal. But we can safely minimize them again.

We can repeat this procedure for every certificate in $D$. The resulting function is not always 1 and, on every 0-input, it has at least $t$ sensitive bits among the bits that
we fixed. Furthermore, we still have $k-t$ non-overlapping valid minimal 1-certificates with no more than 2 contradictions each. In the next section, we show that this
implies that it has 0-sensitivity of at least $k-t$ (Lemma \ref{lemma:graph}). 
Therefore, the original function has a 0-sensitivity of at least $k-t+t=k$.

\subsection{Functions with No Overlaps}

\begin{lemma} \label{lemma:graph}
Let $f$ be a Boolean function, such that $f$ is not always 1 and $f$ has such $k$ non-overlapping minimal 1-certificates that each has at most 2 contradictions with all the others together. Then, $s_0(f) \geq k$.
\end{lemma}

\begin{proof}
To prove this lemma, we consider the weighted graph $G$ on these $k$ certificates where the weight of an edge in this graph is the number of contradictions between the two certificates the edge connects.

We examine the connected components in this graph, not counting edges with weight 0. There can be only 4 kinds of components -- individual certificates, two certificates with 2 contradictions between them, paths of 2 or more certificates with 1 contradiction between every two subsequent certificates in the path
and cycles of 3 or more certificates with 1 contradiction between every two subsequent certificates in the cycle. As there are no overlaps between the certificates, each position is assigned to by certificates from at most one component.

We will now prove by induction on $k$ that we can obtain a 0-input with as many sensitive bits in each component as there are certificates in it.

As a basis we take $k=0$. Since $f$ is not always 1, $s_0(f)$ is defined, but obviously $s_0(f) \geq 0$. 

Then we look at each graph component type separately.

\subsubsection{Individual Certificates.}

We first examine individual certificates.  Let us denote the examined certificate by $c$ and the set of positions it assigns by $S$. 
We fix all bits of $S$ except for one according to $c$ and we fix the remaining bit of $S$ opposite to $c$. The remaining function cannot be always 1, as otherwise the last bit in $S$ would not be necessary in $c$, but $c$ is minimal. Therefore on 0-inputs the remaining function is also sensitive to this last bit, as flipping it produces a word which satisfies $c$.

Afterwards the remaining certificates might no longer be minimal. In this case we can minimize them. This cannot produce any more contradictions and no certificate can disappear, as the function is not always 1. Therefore the remaining function still satisfies the conditions of this lemma and has $k-1$ minimal 1-certificates, with each certificate having at most 2 contradictions with the others.

Then by induction the remaining function has a 0-sensitivity of $k-1$. Together with the sensitive bit among the fixed ones, we obtain $s_0(f) \geq k$.

\subsubsection{Certificate Paths.}

We can similarly reduce certificate paths. A certificate path is a structure where each certificate has 1 contradiction with the next one and there are no other contradictions.
For example, here is an example of a path of length 3:

\begin{equation}
\begin{pmatrix}
&&i&&&& \\
1&1&0&*&*&*&*\\
*&*&1&1&0&*&*\\
*&*&*&*&1&1&1
\end{pmatrix}.
\end{equation}

We note that every certificate in a path assigns at least 2 positions, otherwise its neighbours would not be minimal.

We then take a certificate $c$ at the start of a path, which is next to a certificate $d$. Let $S$ be the set of positions $c$ assigns. Let $i$ be the position where $c$ and $d$ contradict each other. 

We then fix every bit in $S$ but $i$ according to $c$, and we fix $i$ according to $d$. The remaining function cannot be always 1, as otherwise $i$ would not be necessary in $c$,
but $c$ is minimal. But on 0-inputs the remaining function is also sensitive to $i$ because flipping it produces a word which satisfies $c$.

We note that in the remaining function the rest of $d$ (not all of $d$ was fixed because $d$ assigns at least 2 positions) is still a valid certificate, since it only assigns one of the fixed bits and it was fixed according to $d$. Similarly to the first case we can minimize the remaining certificates and obtain a function with $k-1$ certificates satisfying the lemma conditions.

Then by induction the remaining function has a 0-sensitivity of $k-1$. Together with the sensitive bit $i$, we obtain $s_0(f) \geq k$.

\subsubsection{Two Certificates with Two Contradictions.}

Let us denote these 2 certificates as $c$ and $d$ and the two positions where they contradict as $i$ and $j$. 
For example, we can have 2 certificates like this:

\begin{equation}
\begin{pmatrix}
&&i&j& \\
1&1&1&0&*\\
*&*&0&1&1
\end{pmatrix}.
\end{equation}

Let $S$ be the set of positions $c$ assigns and $T$ be the set of positions $d$ assigns.
We then fix every bit in $S$ except $j$ according to $c$ but we fix $j$ according to $d$. 
The remaining function cannot be always 1 because, otherwise, $j$ would not be necessary in $c$, but $c$ is minimal. But on 0-inputs the remaining function is also sensitive to $j$, as flipping it produces a word which satisfies $c$.

If $|T|=2$, then on 0-inputs the remaining function is also sensitive to $i$ because flipping the $i^{\rm th}$ variable produces a word which satisfies $d$. 

If $|T|>2$, we examine the $2^{|T|-2}$ subfunctions obtainable by fixing the remaining positions of $T$. We can w.l.o.g. assume that $d$ assigns the value 1 to each of these. Similarly to section \ref{section:overlaps}, we find the largest non-constant subfunction among these -- the subfunction that is not identically 1 with the highest number of bits fixed to 1. Then on 0-inputs we obtain a sensitive bit either at $i$ if this subfunction fixes all these positions to 1 or at a fixed 0 otherwise.

Therefore we can always find at least one additional sensitive bit among $T$.

Again we can minimize the remaining certificates and obtain a function with $k-2$ certificates satisfying the conditions of the lemma. 

Then by induction the remaining function has a 0-sensitivity of $k-2$. Together with the two additional sensitive bits found, we obtain $s_0(f) \geq k$.

\subsubsection{Certificate Cycles.}
A certificate cycle is a sequence of at least 3 certificates where each certificate has 1 contradiction with the next one and the last one has 1 contradiction
with the first one.
For example, here is a cycle of length 5:
\begin{equation}
\begin{pmatrix}
 j_{5,1}& &j_{1,2}&j_{2,3}& &j_{3,4}&j_{4,5} & \\
1&1&0&*&*&*&*&*\\
*&*&1&0&*&*&*&*\\
*&*&*&1&1&1&*&*\\
*&*&*&*&*&0&0&*\\
0&*&*&*&*&*&1&1
\end{pmatrix}.
\end{equation}

Every certificate in a cycle assigns at least 2 positions, otherwise its neighbours in the cycle would overlap.
We denote the length of the cycle by $m$. Let $c_1, \ldots, c_m$ be the certificates in
this cycle, let $S_1, \ldots, S_m$ be the positions assigned by them, and let $j_{1,2}, \ldots, j_{m,1}$ be the positions where the certificates contradict.

We assign values to variables in $c_2, \ldots, c_m$ in the following way. We first assign values to variables in $S_2$ so that the variable $j_{2,3}$ contradicts $c_2$ and is assigned according to $c_3$, but all other variables are assigned according to $c_2$. 

We have the following properties. First, the remaining function cannot be always 1, as otherwise $j_{2,3}$ would not be necessary in $c_2$, but $c_2$ is minimal. Second, any 0-input that is consistent with the assignment that we made is sensitive to $j_{2,3}$ because flipping this position produces a word which satisfies $c_2$. Third, in the remaining function $c_3$, $\ldots$, $c_m$ are still valid 1-certificates because we have not made any assignments that 
contradict them. Some of these certificates $c_i$ may no longer be minimal. In this case, we can minimize them by removing unnecessary variables from $c_i$ and $S_i$.

We then perform a similar procedure for $c_i \in\{3, \ldots, m\}$. We assume that the variables in $S_2$, $\ldots$, $S_{i-1}$ have been assigned values. We then assign values 
to variables in $S_i$. If $c_i$ and $c_{i+1}$ contradict in the variable $j_{i,i+1}$, we assign it according to $c_{i+1}$. (If $i=m$, we define $i+1=1$.) If $c_i$ and $c_{i+1}$ no longer contradict (this can happen if $j_{i,i+1}$ was removed from one of them), we choose a variable in $S_{i}$ arbitrarily and assign it opposite to $c_i$. All other variables in $S_i$ are assigned according to $c_i$.

We now have similar properties as before. The remaining function cannot be always 1 and any 0-input that is consistent with our assignment is sensitive to changing a variable in $S_i$. Moreover, $c_{i+1}, \ldots, c_m$ are still valid 1-certificates and, if they are not minimal, they can be made minimal by removing variables.

At the end of this process, we have obtained $m-1$ sensitive bits on 0-inputs: for each of $c_2$, $\ldots$, $c_m$, there is a bit, changing which results in an input satisfying
$c_i$. We now argue that there should be one more sensitive bit. To find it, we consider the certificate $c_1$.

During the process described above, the position $j_{1,2}$ where $c_1$ and $c_2$ contradict was fixed opposite to the value assigned by $c_1$. The position $j_{m,1}$ where $c_1$ and $c_m$ contradict is either unfixed or fixed according to $c_1$. All other positions of $c_1$ are unfixed.

If there are no unfixed positions of $c_1$, then changing the position $j_{1,2}$ in a 0-input (that satisfies the partial assignment that we made) leads to a 1-input that satisfies $c_1$.
Hence, we have $m$ sensitive bits.

Otherwise, let $T \subset S_1$ be the set of positions in $c_1$ that have not been assigned and let $p=|T|$. W.l.o.g, we assume that $c_1$ assigns the value 1 to each of those positions. We examine the $2^{p}$ subfunctions obtainable by fixing the positions of $T$ in some way. Again we find the largest non-constant subfunction among these -- the subfunction that is not identically 1 with the highest number of bits fixed to 1. Then on 0-inputs we obtain a sensitive bit either at $j_{1,2}$ if this subfunction fixes all these positions to 1 or at a fixed 0 otherwise.

Similarly to the first three cases, we can minimize the remaining certificates and obtain a function with $k-m$ certificates satisfying the conditions of the lemma. 
By induction, the remaining function has a 0-sensitivity of $k-m$. Together with the $m$ additional sensitive bits we found, we obtain $s_0(f) \geq k$. \qed
\end{proof}

\section{Conclusions}

In this paper, we have shown a lower bound on 1-certificate complexity in relation to the ratio of 0-block sensitivity and 0-sensitivity:
\begin{equation}
\label{eq:final}
C_1(f) \geq \frac{3}{2} \frac{bs_0(f)}{s_0(f)} - \frac{1}{2}.
\end{equation}
This bound is tight, as the function constructed in Theorem \ref{thm:easy} achieves the following equality:
\begin{equation}
\label{eq:easy}
 C_1(f) =  \frac{3}{2} \frac{bs_0(f)}{s_0(f)} + \frac{1}{2}.
\end{equation}
The difference of $1$ appears as the proof of Theorem \ref{theorem:result} requires only a single $1$ in each certificate but the construction of Theorem \ref{thm:easy} has two.
 
Thus, we have completely solved the problem of finding the optimal relationship between 
$s_0(f)$, $bs_0(f)$ and $C_1(f)$.
For functions with $s_1(f)=C_1(f)$, such as those constructed in \cite{AS,R,V}, this means that 
\begin{equation}\label{equation:conjecture}
bs_0(f) \leq \left( \frac{2}{3} +o(1) \right) s_0(f) s_1(f).
\end{equation}
That is, if we use such functions, there is no better separation between $s(f)$ and $bs(f)$
than the currently known one.

For the general case, it is important to understand how big the gap between $s_1(f)$ and $C_1(f)$ can be. Currently, we only know that
\begin{equation}
\label{eq:as} s_1(f) \leq C_1(f) \leq 2^{s_0(f)-1} s_1(f),
\end{equation}
with the upper bound shown in \cite{A+}. In the general case  (\ref{eq:final})  together with this bound implies only
\begin{equation}
bs_0(f) \leq \left( \frac{2}{3} +o(1) \right) 2^{s_0(f)-1} s_0(f) s_1(f).
\end{equation}

However, there is no known $f$ that comes even close to saturating the
upper bound of (\ref{eq:as}) and we suspect that this bound can be significantly improved.

There are some examples of $f$ with gaps between $C_1(f)$ and $s_1(f)$, though. For example,
the 4-bit non-equality function of \cite{A} has $s_0(NE)=s_1(NE)=2$ and $C_1(NE)=3$ and 
it is easy to use it to produce an example $s_0(NE)=2$, $s_1(NE)=2k$ and $C_1(NE)=3k$.
Unfortunately, we have not been able to combine this function with the function
that achieves (\ref{eq:easy}) to obtain a bigger gap between $bs(f)$ and $s(f)$.

Because of that, we conjecture that (\ref{equation:conjecture}) might actually be
optimal. Proving or disproving this conjecture is a very challenging problem.

\bibliographystyle{abbrv}
\bibliography{bibliography}

\end{document}